\documentclass[pra,amsmath,amsfonts,superscriptaddress,showpacs]{revtex4-1}

\usepackage{epsf,epsfig}
\usepackage[psamsfonts]{amssymb}
\usepackage{amsmath}
\usepackage{graphicx}
\usepackage{color,soul}
\usepackage{xcolor}

\newtheorem{theorem}{Theorem}
\newtheorem{lemma}[theorem]{Lemma}
\newtheorem{proposition}[theorem]{Proposition}
\newtheorem{corollary}[theorem]{Corollary}
\newtheorem{remark}[theorem]{Remark}
\newenvironment{proof}[1][Proof]{\begin{trivlist}
\item[\hskip \labelsep {\bfseries #1}]}{\end{trivlist}}

 \newcommand{\ket}[1]{|#1\rangle}
 \newcommand{\bra}[1]{\langle #1|}
 \newcommand{\project}[1]{\ket{#1}\bra{#1}}

 \newcommand{\Id}{{\mathbb I}}
 
 \newcommand{\rank}{{\mathrm rank}}
 \newcommand{\R}{{\mathrm R}}
 \newcommand{\T}{{\mathrm t}}

\newcommand{\Tr}{{\mathrm {Tr}}}
\newcommand{\diag}{{\mathrm {diag}}}

\newcommand{\etal}{\textit {et al.} }

 \newcommand{\half}{\frac{1}{2}}

\begin{document}
\title{Computable  measure of the  quantum correlation}
\author{S. Javad Akhtarshenas\footnote{akhtarshenas@um.ac.ir}}
\affiliation{Department of Physics, University of Isfahan,
 Isfahan, Iran}
\affiliation{Quantum Optics Group, University of Isfahan,
 Isfahan, Iran}
 \affiliation{Department of Physics, Ferdowsi University of Mashhad,
 Mashhad, Iran}
\author{Hamidreza Mohammadi\footnote{h.mohammadi@sci.ui.ac.ir}}\affiliation{Department of Physics, University of Isfahan,
 Isfahan, Iran}
\affiliation{Quantum Optics Group, University of Isfahan,
 Isfahan, Iran}
\author{Saman Karimi\footnote{s.karimi@sci.ui.ac.ir}}\affiliation{Department of Physics, University of Isfahan,
 Isfahan, Iran}
\author{Zahra Azmi\footnote{z.azmi@sci.ui.ac.ir}}
\affiliation{Department of Physics, University of Isfahan,
 Isfahan, Iran}

\begin{abstract}
 A general state of an $m\otimes n$ system is a classical-quantum state if and only if its associated  $A$-correlation matrix (a matrix constructed from the coherence vector of the party $A$,  the correlation matrix of the state, and a function of the local coherence vector  of the subsystem $B$),   has rank no larger than $m-1$.   Using the general Schatten $p$-norms, we quantify  quantum correlation by measuring any violation of this condition.   The required minimization can be carried out for the general $p$-norms and any function of the local coherence vector  of the unmeasured subsystem, leading to  a class of computable quantities which can be used to capture the quantumness of  correlations due to the  subsystem $A$.   We introduce two special members of these quantifiers; The first one  coincides with the tight lower bound on the geometric measure of discord, so that  such  lower bound fully captures the quantum correlation of a bipartite system. Accordingly, a vanishing tight lower bound on the geometric discord is a necessary and sufficient condition for a state to be zero-discord. The second quantifier has the property that it is invariant under a local and reversible operation performed on the unmeasured subsystem, so that it can be regarded as a computable well-defined measure of the quantum correlations. The approach presented in this paper provides a way to circumvent the problem with the geometric discord. We provide some examples to exemplify this measure.
\end{abstract}


\maketitle


\section{Introduction}
Quantum discord represents a new type of quantum correlation which looks at the  correlations from a new perspective, i.e. measurement theory, different from the entanglement-separability paradigm \cite{Zurek2001,Henderson2001}.  The idea of quantum discord is based on the fact that while in the classical physics  measurements can be carried out without disturbance, in quantum mechanics measurements often disturb the system and the
disturbance can be exploited to quantify the quantumness of correlations therein \cite{LuoPRA2008-MIN,LuoFuPRL2011}.
Hence, there exist separable (not-entangled) states which have non-zero discord \cite{LiLuoPRA2008} such that one can employ these separable states as a resource to enhance the quality of quantum information and computation processing
\cite{Cubitt2003,DakicNature2012}. For instance, the  deterministic quantum computation with one qubit  demonstrates such a
speedup without entanglement \cite{Knill1998}. An operational interpretation of quantum discord in terms of state merging is proposed in \cite{Madhok2011,Cavalcanti2011}.  Nowadays, quantum discord became a subject of intensive study in different contexts \cite{Modi2012,LuoSun2010}  and different versions of quantum discord and their measures have been introduced and analyzed \cite{Modi2012,Brodutch2010}. Since the evaluation of quantum discord involves an optimization procedure, almost all quantum discord  measures are very difficult to calculate analytically and quantum discord  was analytically computed only for a few families of two-qubit states \cite{Luo2008, Dillenschneider2008, Sarandy2009, Ali2010, Adesso2010, Giorda2010}, some reduced two-qubit states of pure three-qubit states,  and  a class of rank-2 mixed state of $4\otimes 2$ systems \cite{Cen2011}. Among the various measures of quantum discord, the geometric discord, has been firstly
proposed  by  Dakic {\it et al.}, is a simple and intuitive quantifier of general non-classical correlations \cite{Dakic2010}. Geometric discord  is defined as the squared Hilbert-Schmidt distance between the state of the quantum system and the closest zero-discord state. For a bipartite state  $\rho$ on $\mathcal{H}^A \otimes \mathcal{H}^B$, with $\dim{\mathcal{H}^{A}}=m$ and $\dim{\mathcal{H}^{B}}=n$, the geometric discord  is defined by \cite{Dakic2010}
\begin{equation} \label{GD-Dakic}
D_G(\rho)= \underset{\chi\in\Omega_{0}}{\min} \|\rho-\chi\|^2,
\end{equation}
where  $\Omega_0$ denotes the set of all zero-discord states and $ \|X-Y\|^2=\Tr(X-Y)^2 $ is the 2-norm or square norm in the Hilbert-Schmidt space. This quantity vanishes on the classical-quantum states. It is shown that the geometric discord has an operational interpretation in terms of the average fidelity of the remote state preparation protocol for two-qubit states \cite{DakicNature2012,Giorgi2013}. Dakic \etal also obtained a closed formula for the geometric discord of an arbitrary two-qubit state in terms of coherence vectors and correlation matrix of the state.
Furthermore, an exact expression for the pure $m\otimes m$ states and arbitrary $2\otimes n$ states are obtained \cite{LuoFu2010,LuoFuPRL2011,LuoFu2012}.

An alternative equivalent form for the geometric discord is introduced by Luo and Fu \cite{LuoFu2010}
\begin{equation}\label{GD-Luo2}
D_G(\rho)= \underset{\Pi^A}{\min} \|\rho-\Pi^A(\rho)\|^2,
\end{equation}
where the minimum is taken over all von Neumann measurements $\Pi^A=\{\Pi^A_k\}_{k=1}^m$ on ${\mathcal H}^A$, and $\Pi^A(\rho)=\sum_{k=1}^m(\Pi^A_k\otimes \Id)\rho(\Pi^A_k\otimes \Id)$ with $\mathbb{I}$ as the identity operator on the appropriate space. They have also shown that Eq. (\ref{GD-Luo2}) is equivalent to \cite{LuoFu2010}
\begin{eqnarray}\label{GD-Luo1}
D_G(\rho)=\Tr{(CC^{\T})}-\max_{A}\Tr{(ACC^{\T}A^{\T})},
\end{eqnarray}
where $\T$ denotes transpose,  and $C=(c_{ij})$ is an $m^2\times n^2$-dimensional  matrix defined by
\begin{eqnarray}\label{Rho-Bipart1}
\rho=\sum_{i=0}^{m^2-1}\sum_{j=0}^{n^2-1} c_{ij}X_i\otimes Y_j,
\end{eqnarray}
with $\{X_i\}_{i=1}^{m^2-1}$ and $\{Y_j\}_{j=1}^{n^2-1}$ as the sets of Hermitian operators which constitute orthonormal basis for $SU(m)$ and $SU(n)$ algebra, respectively, i.e.
\begin{eqnarray}
\Tr(X_iX_{i^\prime})=\delta_{i i^\prime},\qquad \Tr(Y_j Y_{j^\prime})=\delta_{j j^\prime}.
\end{eqnarray}
In Eq. (\ref{GD-Luo1}), the maximum is taken over all $m\times m^2$-dimensional matrices $A=(a_{ki})$ such that
\begin{eqnarray}\label{aki}
a_{ki}=\Tr{(\ket{k}\bra{k}X_i)}=\bra{k}X_i\ket{k},
\end{eqnarray}
where $\{\ket{k}\}_{k=1}^{m}$ is any orthonormal base for $\mathcal{H}^{A}$.

For further use, let us give another useful  representation for  a general  bipartite state
$\rho$ on $\mathcal{H}^{A}\otimes \mathcal{H}^{B}$  as
\begin{eqnarray}\label{Rho-Bipart2}
\rho=\frac{1}{mn}\left(\Id\otimes \Id +\vec{x}\cdot\hat{\lambda}^A\otimes \Id+\Id\otimes {\vec y}\cdot\hat{\lambda}^B
+\sum_{i=1}^{m^2-1}\sum_{j=1}^{n^2-1}t_{ij} \hat{\lambda_{i}^A}\otimes \hat{\lambda}_{j}^B\right),
\end{eqnarray}
where
 $\{\hat{\lambda}_i^A\}_{i=1}^{m^2-1}$ and  $\{\hat{\lambda}_j^B\}_{j=1}^{n^2-1}$ are generators of $SU(m)$ and $SU(n)$, respectively, fulfilling the following relations
\begin{eqnarray}\label{SUmGellMann}
 \Tr{\hat{\lambda}_i^s}=0,\qquad \Tr(\hat{\lambda}_i^s\hat{\lambda}_j^s)=2\delta_{ij}, \qquad s=A,B.
\end{eqnarray}
Here $\Id$ stands for the identity operator, $\vec{x}=(x_1,\cdots,x_{m^2-1})^{\T}$ and $\vec{y}=(y_1,\cdots,y_{n^2-1})^{\T}$ are local coherence vectors of the subsystems $A$ and $B$, respectively
\begin{eqnarray}
x_i &=&\frac{m}{2}\Tr{\left[(\hat{\lambda}_{i}^A\otimes \Id)\rho\right]},\quad
y_j =\frac{n}{2}\Tr{\left[(\Id\otimes \hat{\lambda}_{j}^B)\rho\right]},
\end{eqnarray}
and $T=(t_{ij})$ is the correlation matrix
\begin{eqnarray}
t_{ij}=\frac{mn}{4}\Tr{\left[(\hat{\lambda}_{i}^A\otimes \hat{\lambda}_{j}^B)\rho\right]}.
\end{eqnarray}
The two representations (\ref{Rho-Bipart1}) and (\ref{Rho-Bipart2}) of $\rho$ are related to each by the following relation \cite{Rana2012,Hassan2012}
\begin{eqnarray}\label{C}
C=\frac{1}{\sqrt{mn}}
\left(\begin{array}{cc} 1 & \sqrt{\frac{2}{n}}\vec{y}^{\T} \\
\sqrt{\frac{2}{m}}\vec{x} & \frac{2}{\sqrt{mn}}T
\end{array}
\right).
\end{eqnarray}
Based on the definition (\ref{GD-Luo1}),  Rana \etal \cite{Rana2012} and Hassan \etal \cite{Hassan2012} have obtained a tight lower bound on the geometric discord  as
\begin{eqnarray}\label{TightLowerBound}
D_G(\rho)\ge \frac{2}{m^2 n}\left(\|\vec{x}\|^2+\frac{2}{n}\|T\|^2- \sum_{k=1}^{m-1}\eta_{k}^\downarrow\right)
= \frac{2}{m^2 n} \sum_{k=m}^{m^2-1}\eta_{k}^\downarrow,
\end{eqnarray}
where  $\{\eta_{k}^\downarrow\}_{k=1}^{m^2-1}$ are eigenvalues of \cite{Rana2012,Hassan2012}
\begin{equation}\label{G}
G:=\vec{x}\vec{x}^{\T}+\frac{2}{n}TT^{\T},
\end{equation}
in nonincreasing order. Remarkably, the above lower bound on the geometric discord is tight in the sense that
for $m\otimes m$ Werner and isotropic states, the above lower bound are achieved \cite{LuoFu2010,Rana2012}.  Furthermore, for an arbitrary state of $2\otimes n$ systems, the geometric discord coincides with this lower bound \cite{LuoFuPRL2011,Saj2012}.

Using Eq. (\ref{GD-Luo1}), it may be interesting to mention here that one can also write geometric discord (\ref{GD-Dakic}) in the following equivalent form
\begin{eqnarray}\label{GD-AlternativeForm2-0}
D_{G}(\rho)=\frac{2}{m^2 n}\left[\Tr{G}-\max_{\{\vec{\mu}_k\}} \sum_{k=1}^{m}\vec{\mu}_k^{\T} G \vec{\mu}_k\right],
\end{eqnarray}
where $G$ is defined by Eq. (\ref{G}) and  maximum is taken over all simplexes $\Delta^{m-1}_{\{\vec{\mu}_k\}\in \mathbb{R}^{m^2-1}}$, i.e. all vectors $\{\vec{\mu}_k\}_{k=1}^{m}\in \mathbb{R}^{m^2-1}$ fulfilling conditions
 $\vec{\mu}_k\cdot\vec{\mu}_{k^\prime}=\left(\delta_{kk^\prime}-\frac{1}{m}\right)$ and $\sum_{k=1}^{m}\vec{\mu}_k=\vec{0}$. For a proof of Eq. (\ref{GD-AlternativeForm2-0})  see Appendix \ref{AppendixTight}. For $m=2$, Eq. (\ref{GD-AlternativeForm2-0}) immediately leads to the geometric discord of $2\otimes n$ states.

As it is clear from the above discussion, the most important future of geometric discord is its computability which is appreciated for use of the Hilbert-Schmidt metric as a measure of distance. However, as it is pointed out by Piani \cite{Piani2012}, geometric discord may increase under local operations on the unmeasured subsystem, so  it can  not be the best conceptual and operational choice to quantify the
quantumness of correlations. The source of this problem can be identified in the fact that the geometric discord  is based on the Hilbert-Schmidt norm that
is noncontractive under trace preserving quantum channels \cite{Piani2012}.
In order to fix this problem Piani has proposed to redefine the geometric discord as  ${\tilde D}_G(\rho)=\sup_{\Lambda_B}D_G(\Lambda_B(\rho))$, with the supremum over all the local channels $\Lambda_B$ on the part B. In view of this, the geometric discord may be interpreted as a lower bound to ${\tilde D}_G(\rho)$.

 On the other hand, Paula $\etal$ have considered the general Schatten $p$-norms and have shown that the 1-norm is the only $p$-norm able to define a consistent quantum correlation measure \cite{Paula2013}. Furthermore, by restricting the optimization to the tetrahedral of two-qubit Bell-diagonal states, they have also obtained an analytical expression for the 1-norm geometric discord of a general two-qubit Bell-diagonal state. Further results on the analytical calculations of the 1-norm geometric discord is presented in \cite{Ciccarello2014}, where the authors have obtained the analytical expressions for the 1-norm geometric discord for a class of two-qubit states including  quantum-classical states and X states.
 Based on the relative entropy \cite{ModiPRL2010}, Hilbert-Schmidt norm \cite{BellomoPRA2012}, and  trace distance \cite{AaronsonNJP2013}, a  unified view of quantum, classical, and total correlations in the bipartite quantum systems is given.
In a different approach, Tufarelli $\etal$ have defined a rescaled version of the geometric discord \cite{Tufarelli2013}, and have shown that the rescaled discord is obtained by  renormalizing the original geometric discord by the purity of state.  However they have pointed out that although the new measure prevents quantum correlation measure from being biased by the global purity of the state \cite{Piani2012}, it still inherits from the original geometric discord  the noncontractive behavior under quantum operations on the unmeasured subsystem, so that it can be regarded as an indicator rather than as a well-behaved measure of quantum correlation.  In \cite{ChangLuo2013}, Chang and Luo  have shown that the problem with the geometric discord can be remedied simply by starting from the square root of a density operator, rather than the density operator itself, in defining the discord. They have derived the analytical formulas for any pure state and any $2\otimes n$ state. Spehner and Orszag \cite{SpehnerNJP2013,SpehnerJPA2014} have used the Bures distance and introduced a distance-based quantum discord. They have shown that for pure states it is identical to the geometric measure of entanglement and for mixed states it coincides with the optimal success probability of an unambiguous quantum state discrimination task \cite{SpehnerNJP2013}.   They have also derived an explicit formula  for the Bell-diagonal states \cite{SpehnerJPA2014}.

Based on the rank of the correlation matrix, Dakic \etal \cite{Dakic2010} obtained a simple necessary condition for a general bipartite state to be zero-discord. A necessary and sufficient condition for a two-qubit state to be zero-discord is obtained by Lu \etal \cite{Lu2011}. Their condition is related to the existence of a unit vector $\hat{n}\in{\mathbb R}^3$ satisfying the following conditions
\begin{equation}\label{ZDS-n-xT}
\hat{n}\hat{n}^{\T}\vec{x}=\vec{x},\qquad \hat{n}\hat{n}^{\T}T=T,
\end{equation}
where $\vec{x}$ denotes coherence vector of the subsystem $A$, and $T$ is the correlation matrix of $\rho$ in Bloch representation. Accordingly, a two-qubit state is of zero-discord if and only if  either $T=0$, or  $\rank(T)=1$ and $\vec{x}$ belongs to the range of $T$. Recently  Zhou \etal \cite{Zhou2011}, based on the extended version of Eq. (\ref{ZDS-n-xT}) (see Eq. (\ref{ZDS-Pi-xT}) below), introduced a criterion tensor as
\begin{equation}\label{C-T}
\Lambda=\left(\frac{4}{mn}\right)^2\left(TT^{\T}-y^2\vec{x}\vec{x}^{\T}\right),
\end{equation}
and showed that a necessary and sufficient condition for a bipartite state to be zero-discord is $\rank(\Lambda)\le m-1$.
Based on this criterion tensor, the authors of \cite{Zhou2011}  proposed a measure of the quantum correlation
as
\begin{equation}\label{Q-Rho}
Q_{\Lambda}(\rho)=\frac{1}{4}\sum_{k=m}^{m^2-1}|\Lambda^{\downarrow}_{k}|,
\end{equation}
where $\{\Lambda^{\downarrow}_{k}\}_{k=1}^{m^2-1}$ are eigenvalues of the criterion tensor (\ref{C-T}) in nonincreasing order.
They have also shown  that in some particular cases their measure coincides with the geometric measure of quantum discord.

In this paper we use the notion of the $A$-correlation matrix and propose a geometric way of quantifying quantum correlation. The optimization involved in the definition can be carried out analytically for the general Schatten $p$-norms and an arbitrary function of the local coherence vector of the unmeasured subsystem, leading therefore to a class of closed form for the quantumness of correlation.  Remarkably, this class of quantifier includes the tight lower bound on the geometric discord given in (\ref{TightLowerBound}). This suggest that such lower bound fully captures the quantum correlation and may be used as an indicator of the quantum correlation. On the other hand, we show that this class of computable quantifier includes  a measure of the quantum correlation  invariant under local  quantum channels performing on the unmeasured part.  In view of this we show that a way to circumvent the issue arisen by Piani is to rescale the original geometric discord just by dividing it by the purity of the unmeasured part.

The paper is organized as follows.
In section II, we review some properties of coherence vectors of an arbitrary  set of von Neumann projection operators on $\mathcal{H}^A$. The necessary and sufficient condition for a state to be zero-discord is also given in section II. Section III is devoted to the definition of the new measure of quantumness. In this section we also present some properties of the new measure and provide a comparison of this measure with the geometric measure and the measure given in Ref. \cite{Zhou2011}. The paper is concluded in section IV.

\section{Characterizing classical-quantum  states}
A general density operator on ${\mathcal H}^A$ can be written as
\begin{eqnarray}\label{CVR}
\rho^A=\frac{1}{m}\left(\Id+\vec{x}\cdot\hat{\lambda}^A\right),
\end{eqnarray}
where $(m^2-1)$-dimensional vector $\vec{x}=(x_1,\cdots, x_{m^2-1})^{\T}$, with $x_i=\frac{m}{2}\Tr{(\hat{\lambda}_i^A\rho^A)}$, is the so-called coherence vector of $\rho^A$.
For further use, we give bellow some properties of coherence vectors of a set of orthonormal pure states. Let $\{\ket{k}\}_{k=1}^{m}$ be an arbitrary orthonormal base for ${\mathcal H}^A$ and $\{\Pi_k^A=\ket{k}\bra{k}\}_{k=1}^{m}$ denotes projectors on this base; then
\begin{eqnarray}\label{Ortho-Project}
\Pi_k^A\Pi_{k^\prime}^A=\Pi_k^A\delta_{kk^\prime},\qquad \sum_{k=1}^{m}\Pi_k^A=\Id.
\end{eqnarray}
Now let $\vec{\alpha}_k\in \mathbb{R}^{m^2-1}$ denotes  coherence vector corresponding to $\Pi_k^A$, i.e.
\begin{eqnarray}\label{CVR-Pure1}
\Pi_k^A=\frac{1}{m}\left(\Id+\vec{\alpha}_k\cdot\hat{\lambda}^A\right),
\end{eqnarray}
then the orthonormality and completeness conditions given in Eq. (\ref{Ortho-Project}) require that $\{\vec{\alpha}_k\}_{k=1}^m$  fulfill the following two conditions
\begin{eqnarray}\label{CVR-Pure2}
\vec{\alpha}_k\cdot\vec{\alpha}_{k^\prime}=-\frac{m}{2}+\frac{m^2}{2}\delta_{kk^\prime},\qquad \sum_{k=1}^{m}\vec{\alpha}_k=\vec{0}.
\end{eqnarray}
From the first relation above we find
\begin{eqnarray}\label{CVR-Pure3}
|\vec{\alpha}_k|=\sqrt{\frac{m(m-1)}{2}}, \qquad \cos{\theta_{k k^\prime}}={\frac{-1}{m-1}},
\end{eqnarray}
where $\theta_{k k^\prime}$ ($k\ne k^\prime$) is the angle between a pair of coherence vectors $\vec{\alpha}_k$ and $\vec{\alpha}_{k^\prime}$.
This implies that the set of $(m^2-1)$-dimensional coherence vectors $\{\vec{\alpha}_{k}\}_{k=1}^{m}$ corresponding to an orthonormal base
forms an $(m-1)$-dimensional simplex. In what follows, we denote this kind of simplex by  $\Delta^{m-1}_{\{\vec{\alpha}_k\}\in \mathbb{R}^{m^2-1}}$.
Corresponding to any such defined simplex, the following lemma gives an $(m-1)$-dimensional projection operator on $\mathbb{R}^{m^2-1}$ \cite{Zhou2011}.
\begin{lemma}\label{LemmaProjection}
Any $(m-1)$-dimensional projection operator on $(m^2-1)$-dimensional space $\mathbb{R}^{m^2-1}$ can be represented by
\begin{eqnarray}\label{Projection}
{P}=\frac{2}{m^2}\sum_{k=1}^m \vec{\alpha}_{k} ({\vec{\alpha}_{k}})^{\T},
\end{eqnarray}
where $\{\vec{\alpha}_{k}\}_{k=1}^{m}$ are coherence vectors corresponding to orthonormal projections, satisfying  Eqs. (\ref{CVR-Pure2}) and (\ref{CVR-Pure3}).
\end{lemma}
\begin{proof} First note that one can easily show that ${P}^\dag={P}$ and ${P}^2={P}$, so  ${P}$ is a projection operator.
Since coherence vectors corresponding to orthonormal base make simplex $\Delta^{m-1}_{\{\vec{\alpha}_{k}\}\in \mathbb{R}^{m^2-1}}$,
so  ${P}$ is an $(m-1)$-dimensional projection operator on $\mathbb{R}^{m^2-1}$ or equivalently it is the unit operator on
space $\mathbb{R}^{m-1}$.
\end{proof}

Let us turn our attention on the bipartite state $\rho$ on $\mathcal{H}^{A}\otimes \mathcal{H}^{B}$ and consider the set of zero-discord states.
By definition, a bipartite state $\rho$ is of  zero-discord, i.e. classical-quantum state,  if and only if there exists orthonormal base $\{\ket{k}\}_{k=1}^{m}$ of $\mathcal{H}^A$ such that \cite{Zurek2001}
\begin{eqnarray}\label{ZDS}
\rho=\sum_{k=1}^{m}p_k\Pi_k^A\otimes \rho_k^B,
\end{eqnarray}
where $\Pi_k^A=\ket{k}\bra{k}$ and $\rho_k^B$ is a state on $\mathcal{H}^{B}$.  The following theorem gives a criterion for a state to be zero-discord \cite{Zhou2011}.
\begin{theorem}\label{TheoremZDS-NC}
A bipartite state $\rho$ on the $\mathcal{H}^A\otimes \mathcal{H}^B$ is a zero-discord state, a classical-quantum state,  if and only if there exists an $(m-1)$-dimensional projection operator ${P}$ on the $(m^2-1)$-dimensional space $\mathbb{R}^{m^2-1}$ such that
\begin{eqnarray}\label{ZDS-Pi-xT}
 {P} \vec{x}=\vec{x}, \qquad {P} T=T,
\end{eqnarray}
where $\vec{x}$ denotes coherence vector of party $A$, and $T$ is the correlation matrix of $\rho$.
\end{theorem}
A proof of this theorem is given in Appendix \ref{AppendixProof} (see also \cite{Zhou2011}). Let us mention here that conditions (\ref{ZDS-Pi-xT}) can be written also as
\begin{eqnarray}\label{ZDS-Pi-Tau}
 {P} {\mathcal T}={\mathcal T},
\end{eqnarray}
where ${\mathcal T}$ is an $(m^2-1)\times n^2$ matrix, obtained by removing the first row of the $m^2\times n^2$ matrix $C$ of Eq. (\ref{C}), i.e.
\begin{equation}\label{LeftCorrelation}
{\mathcal T}:=\sqrt{\frac{2}{m^2 n}}\left(\begin{array}{cc}
\vec{x}\; & \;\sqrt{\frac{2}{n}}T
\end{array}\right).
\end{equation}
Since $\mathcal{T}$ includes coherence vector $\vec{x}$ of the subsystem  $A$ as well as the correlation matrix $T$ of the bipartite system $A-B$,  we call $\mathcal{T}$ as the $A$-correlation matrix associated to the state $\rho$.

As an example, let us consider the case of two-qubit system. In this case a general zero-discord state $\chi$ is characterized  by
$\vec{x}=(p_1-p_2)\hat{n}$, $\vec{y}=(p_1\vec{\xi}_1+p_2\vec{\xi}_2)$, and $T=\hat{n}(p_1\vec{\xi}_1-p_2\vec{\xi}_2)^{\T}$, where $p_1,p_2$ are probabilities with $p_1+p_2=1$, $\hat{n}$ is a unit vector, and $\vec{\xi}_1,\vec{\xi}_2$ are coherence vectors of the subsystem $B$. Evidently, the zero-discord condition (\ref{ZDS-n-xT}) is satisfied.
In the following we show that the above theorem provides a necessary and sufficient  condition for a bipartite state $\rho$ to be zero-discord \cite{Dakic2010,Zhou2011}.
\begin{corollary} A  bipartite state $\rho$ with the $A$-correlation matrix ${\mathcal T}$, associated to the local coherence vector  $\vec{x}$ and  correlation matrix $T$,  is a classical-quantum state, i.e. zero-discord  state,  if and only if
\begin{eqnarray}\label{RangeConditions1}\nonumber
\rank(\mathcal{T}\mathcal{T}^{\T})\le m-1.
\end{eqnarray}
Equivalently, one can say that $\rho$ is a zero-discord state if and only if one of the following conditions is satisfied
\begin{eqnarray}\label{RangeConditions2}\nonumber
(i)&\quad \rank(TT^{\T})\le m-2,\qquad\qquad\qquad\quad\quad\;\;\\ \nonumber
(ii)&\quad \rank(TT^{\T})\le m-1,\quad  \textrm{and} \quad \vec{x}\in \R({TT^{\T}}),
\end{eqnarray}
where $\R(M)$  denotes range of  the matrix $M$.
\end{corollary}
\begin{remark}
Note that the definition (\ref{LeftCorrelation}) for the $A$-correlation matrix is not the most general one. In fact since the  coherence vector $\vec{y}$ of the second subsystem is invariant under any measurement on the first subsystem, the zero-discord condition (\ref{ZDS-Pi-Tau}) is still satisfied if  we  extend the $A$-correlation matrix $\mathcal{T}$
in a more general form as
\begin{equation}\label{LeftCorrelation-ff}
{\mathcal T}_{f}:=\sqrt{\frac{2}{m^2 n}}\left(\begin{array}{cc}
f_1(y)\vec{x}\; & \;\sqrt{\frac{2}{n}}f_2(y)T
\end{array}\right),
\end{equation}
where $f=\{f_1,f_2\}$ with $f_1(y)$ and $f_2(y)$ as  two, in general complex, functions of $y$. However, the definition given by Eq. (\ref{LeftCorrelation})  is unique in the sense that it is constructed from the expansion coefficients  of the density matrix $\rho$ in terms of the orthonormal basis $X_i\otimes Y_j$  of Eq. (\ref{Rho-Bipart1}), i.e. $\mathcal{T}_{ij}=c_{ij}$ for $i=1,\cdots,m^2-1$ and $j=0,\cdots,n^2-1$.
\end{remark}

\section{Quantifying Quantum Correlations}
\subsection{Computable measure of quantum correlation}
Theorem \ref{TheoremZDS-NC} allows us to introduce a new measure of quantum correlation. Since condition (\ref{ZDS-Pi-Tau}) gives necessary and sufficient  condition for a state to be zero-discord, therefore measuring any violation of this   condition can be used to quantify correlation. Accordingly, we use the degree to which the above condition fails to be satisfied as a measure of quantum correlation. Here we propose the following measure of quantum correlation
\begin{proposition} For a given bipartite state $\rho$  with the local coherence vectors   $\vec{x}$, $\vec{y}$  and  the correlation matrix $T$
we propose the  following quantity as a quantum correlation measure
\begin{eqnarray}\label{NewGD}
D_{{{\mathcal T}},f}^{(p)}(\rho)=\min_{{P}}\|{\mathcal T}_f-{P} {\mathcal T}_f\|_{p},
\end{eqnarray}
where $\|A\|_{p}=\left[\Tr{(A^\dag A)^{p/2}}\right]^{1/p}$ (for $p\ge 1$) is the so-called Schatten $p$-norm \cite{HornBook2013}, and the minimum is taken over all $(m-1)$-dimensional projection operators $P$ on $\mathbb{R}^{m^2-1}$. Also the generalized $A$-correlation matrix ${\mathcal T}_f$ is defined by Eq. (\ref{LeftCorrelation-ff}).
\end{proposition}
As we show below, minimization involved in the definition given above can be solved analytically for any $p\ge 1$ and arbitrary choice of the functions $f_1(y)$ and $f_2(y)$, giving therefore a closed form of expression for the quantum correlation of an arbitrary  $m\otimes n$  bipartite state.
To see this, we write
\begin{eqnarray}\nonumber
D_{{\mathcal T},f}^{(p)}(\rho)&=&\min_{{P}}\|{\mathcal T}_f-{P} {\mathcal T}_f\|_{p}=\min_{{P^{\perp}}}\|P^{\perp} {\mathcal T}_f\|_{p} \\ \label{NewGD2}
&=&\min_{{P^{\perp}}}\left[\Tr{(P^{\perp}{\mathcal T}_f{\mathcal T}_f^\dag P^{\perp})^{p/2}}\right]^{1/p},
\end{eqnarray}
where we have defined $P_{\perp}=\Id-P$ as the $m(m-1)$-dimensional projection operator on $\mathbb{R}^{m^2-1}$.
Invoking the fact that for any Hermitian operator $H$ and any projection operator $P$, the  eigenvalues of the restricted matrix $P H P$ lie between the eigenvalues of the matrix $H$, i.e. $\min\textrm{Eig}\{H \}\le \textrm{Eig}\{PHP\}\le \max\textrm{Eig}\{H\}$, we find  the following expression for the quantum correlation of the bipartite state $\rho$
\begin{eqnarray}\label{s-f-measure}
D_{{\mathcal T},f}^{(p)}(\rho)=\left[\sum_{k=m}^{m^2-1}(\tau_{k}^{f\downarrow})^{p/2}\right]^{1/p},
\end{eqnarray}
where  $\{\tau_{k}^{f\downarrow}\}_{k=1}^{m^2-1}$ are eigenvalues of ${\mathcal T}_f{\mathcal T}_f^{\dagger}=\frac{2}{m^2 n}\left(|f_1(y)|^2 \vec{x}\vec{x}^{\T}+\frac{2}{n}|f_2(y)|^2T T^{\T}\right)$ in nonincreasing order.

Before we give properties of the the above measure,
it is worth to mention that Eq. (\ref{s-f-measure}) gives us a closed relation for an arbitrary Schatten $p$-norm, i.e.  any $p\ge 1$. Some  important candidates for $p$ may be: (i) Trace class norm ($p=1$), $D_{{\mathcal T},f}^{(1)}(\rho)=\sum_{k=m}^{m^2-1}\sqrt{\tau_{k}^{f\downarrow}}$. (ii) Hilbert-Schmidt norm ($p=2$), $D_{{\mathcal T},f}^{(2)}(\rho)=\sqrt{\sum_{k=m}^{m^2-1}\tau_{k}^{f\downarrow}}$. (iii) Operator norm ($p\rightarrow \infty$), which for the linear transformation $A: \mathcal{H}\rightarrow \mathcal{H}^\prime$ is defined by $\|A\|_{op}=\max\{|A\psi|\; : \; \psi\in \mathcal{H}, \; |\psi|=1\}$ which is equal to $\|A\|_{\infty}=\lim_{p\rightarrow \infty}\|A\|_{p}$. In this case we have $D_{{\mathcal T},f}^{(\infty)}(\rho)=\sqrt{\tau_{m}^{f\downarrow}}$.
We are now in the position to present some properties of the above measure of quantum correlation.
\begin{enumerate}
\item
 By definition, the above measure of quantum correlation  vanishes only for zero-discord states.
\item
For any maximally entangled state $\ket{\Psi}=\frac{1}{\sqrt{m}}\sum_{i=1}^{m}\ket{ii}$,  we find ${\mathcal T}_f{\mathcal T}_f^{\T}=\frac{|f_2(0)|^2}{m^2}I_{m^2-1}$ (see example below), so that $D_{{\mathcal T},f}^{(p)}(\rho)=\frac{[m(m-1)]^{1/s}}{m}|f_2(0)|$. This, in particular,  achieves its maximum value if $|f_2(y)|$ be a constant or a decreasing function of $y=\sqrt{\vec{y}^{\T}\vec{y}}$.
\item
 $D_{{\mathcal T},f}^{(p)}(\rho)$ is invariant under any local unitary operations $U_1$ and $U_2$  performed on ${\mathcal H}^A$ and $\mathcal{H}^B$  respectively, i.e. $D_{{\mathcal T},f}^{(p)}((U_1\otimes U_2)\rho(U_1\otimes U_2)^\dagger)=D_{{\mathcal T},f}^{(p)}(\rho)$ where $U_1\in SU(m)$ and $U_2\in SU(n)$. This follows from the fact under such transformations, the coherence vectors ${\vec x}$, ${\vec y}$ and the correlation matrix $T$ transform as
\begin{equation}
{\vec x}\rightarrow O_1{\vec x},\qquad {\vec y}\rightarrow O_2{\vec y},\qquad T\rightarrow O_1TO_2^{\T},
\end{equation}
where $O_1$ corresponds to $U_1$ via $U_1({\vec x}\cdot {\hat \lambda}^A)U_1^\dag= (O_1{\vec x})\cdot {\hat \lambda}^A$ with $O_1\in SO(m^2-1)$, and a similar definition holds for $O_2$. This leads to $({\mathcal T}_f{\mathcal T}_f^{\T})\rightarrow O_1 ({\mathcal T}_f{\mathcal T}_f^{\T}) O_1^{\T}$, leaving  eigenvalues of ${\mathcal T}_f{\mathcal T}_f^{\T}$ invariant.
\item
$D_{{\mathcal T},f}^{(p)}(\rho)$ is invariant under local reversible operations on the unmeasured subsystem if we choose  $f_1(y)=f_2(y)=1/\sqrt{\mu{(\rho^B)}}$ where $\mu(\rho^B)=\Tr{(\rho^B)^2}$ is the purity of the subsystem $B$. Explicitly, this means that for any map $\Gamma^C : \rho\rightarrow \rho\otimes \rho^C$, i.e. any channel that introduces a noisy ancillary state $\rho^C$ on the unmeasured subsystem, we have that $D_{{\mathcal T},f}^{(p)}(\Gamma^C(\rho))=D_{{\mathcal T},f}^{(p)}(\rho)$.
 To show this let ${\mathcal T}_f^{AB}$  and ${\mathcal T}_f^{A(BC)}$ be the $A$-correlation matrices  associated to the input and output states $\rho^{AB}$ and $\Gamma^C(\rho^{AB})=\rho^{AB}\otimes \rho^C$, respectively. Using the coherence vector representation of $\rho^C$ as $\rho^C=\frac{1}{n^\prime}\left(\Id+\vec{z}\cdot\hat{\lambda}^C\right)$,
where $\vec{z}=(z_1,\cdots, z_{{n^\prime}^2-1})^{\T}$, with $z_i=\frac{n^\prime}{2}\Tr{(\hat{\lambda}_i^C\rho^C)}$, and $n^{\prime}$ be the dimension of the ancillary Hilbert space, one can shows after some calculations that ${\mathcal T}_f^{A(BC)}=f(z)\frac{1}{\sqrt{n^\prime}}\left(\begin{array}{cc}1 \; &\; \sqrt{\frac{2}{n^\prime}}\vec{z}^\T\end{array}\right)\otimes{\mathcal T}_f^{AB}$. Using the fact that $\mu(\rho^C)=\frac{1}{n^\prime}\left(1+\frac{2}{n^\prime}\vec{z}^\T\vec{z}\right)$, this immediately leads to $({\mathcal T}_f^{A(BC)})({\mathcal T}_f^{A(BC)})^{\T}=[f(z)]^2\mu(\rho^C)({\mathcal T}_f^{AB})({\mathcal T}_f^{AB)})^{\T}=({\mathcal T}_f^{AB})({\mathcal T}_f^{AB)})^{\T}$. It follows therefore that $D_{{\mathcal T},f}^{(p)}(\Gamma^C(\rho^{AB}))=D_{{\mathcal T},f}^{(p)}(\rho^{AB})$.
\end{enumerate}
Properties 1 to 3 show that for any choice  of $f_1(y)$ and $f_2(y)$, the quantity  $D_{{\mathcal T},f}^{(p)}(\rho^{AB})$ can be regarded as a computable  indicator for the quantum correlations of the bipartite state $\rho$ due to the first subsystem. Property 4, however,  indicates that for the unique choice   $f_1(y)=f_2(y)=1/\sqrt{\mu{(\rho^B)}}$, the corresponding quantity have the required property of being   invariant under local  quantum channels performing on the unmeasured part; as such  it can be regarded as a  computable well-defined  measure of quantum correlation.

\begin{proposition} For a given bipartite state $\rho$ with the $A$-correlation matrix ${\mathcal T}_\mu$, associated to the  local coherence vector   $\vec{x}$ of the party $A$,  local  purity $\mu(\rho^B)$ of the party $B$, and  the correlation matrix $T$, we propose the computable well-defined measure of the  quantum correlation as
\begin{eqnarray}
D_{{\mathcal T},\mu}^{(p)}(\rho)=\min_{{P}}\|{\mathcal T}_\mu-{P} {\mathcal T}_\mu\|_{p}=\frac{1}{\sqrt{\mu{(\rho^B)}}}\left[\sum_{k=m}^{m^2-1}(\tau_{k}^{\downarrow})^{p/2}\right]^{1/p}.
\end{eqnarray}
where  $\{\tau_{k}^{\downarrow}\}_{k=1}^{m^2-1}$ are eigenvalues of ${\mathcal T}{\mathcal T}^{\T}=\frac{2}{m^2 n}\left( \vec{x}\vec{x}^{\T}+\frac{2}{n}T T^{\T}\right)$ in nonincreasing order.
\end{proposition}
In the following subsection we present some measures that can be obtained from the general formula (\ref{s-f-measure}).

\subsection{Relation with the other measures}

{\it Geometric discord.---}
It is worth to note that the square of $D_{{\mathcal T},f}^{(2)}(\rho)$ for $f=1$, i.e. for $f=\{f_1=1,f_2=1\}$, is closely related to the geometric discord. It follows from the fact
\begin{equation}
{{\mathcal T}_{f=1}{\mathcal T}_{f=1}^{\T}}=\frac{2}{m^2n}\left(\vec{x}\vec{x}^{\T}+\frac{2}{n}TT^{\T}\right)=\frac{2}{m^2n}G,
\end{equation}
which immediately indicates that  $[D_{{\mathcal T},f=1}^{(2)}(\rho)]^2$ coincides with the tight lower bound  on the geometric discord (\ref{TightLowerBound}), therefore we have in general
\begin{equation}
\left[D_{{\mathcal T},f=1}^{(2)}(\rho)\right]^2\le D_G(\rho),
\end{equation}
where the equality is satisfied when the first subsystem is a qubit. This, particularly,  implies that  for an $m\otimes n$ system, $\rho$ is a zero-discord state if and only if  the lower bound (\ref{TightLowerBound}) vanishes.
Recall that one refers to a bound as faithful if and only if it vanishes on any state for which the bounded quantity vanishes.
In view of this  the tight lower bound given by Eq. (\ref{TightLowerBound}) is faithful, so it may serve as an independent indicator  of the quantumness.
It should be noted that our  measure of quantumness can be regarded as a kind of geometric measure. Indeed, the geometric discord $D_G(\rho)$, as given in  Eq. (\ref{GD-Luo2}), is defined as the square of the Hilbert-Schmidt distance between a given $\rho$ and the closest state $\Pi^A(\rho)$, for all von Neumann (projective) measurements $\Pi^A=\{\Pi^A_k\}_{k=1}^{m}$ acting on $\mathcal{H}^A$. On the other hand,  $D_{{\mathcal T},f}^{(p)}(\rho)$ is defined as the $p$-distance between the $A$-correlation matrix $\mathcal{T}_f$ associated to  $\rho$ and the closest $A$-correlation matrix $P(\mathcal{T}_f)$, for all $(m-1)$-dimensional projection operators  $P$ acting on $\mathbb{R}^{m^2-1}$.

{\it Criterion tensor $\Lambda$.---}
Surprisingly, the general definition (\ref{LeftCorrelation-ff}) of the $A$-correlation matrix enables one to obtain the  criterion tensor $\Lambda$ as well as the nonclassicality $Q_{\Lambda}(\rho)$ of Ref. \cite{Zhou2011}. To this aim, let us choose $f$ as  $f_{\Lambda}=\{f_1(y)=i\sqrt{\frac{2}{n}}y,f_2(y)=1\}$ in Eq. (\ref{LeftCorrelation-ff}), and get  $\frac{1}{4}\Lambda=\mathcal{T}_{f_\Lambda}\mathcal{T}_{f_{\Lambda}}^\T$. Now the nonclassicality $Q_{\Lambda}(\rho)$ can be obtained as \cite{Zhou2011}
\begin{eqnarray}\label{Q-Rho-oneNorm}
Q_{\Lambda}(\rho)=\frac{1}{4}\min_{{P}}\{\|{\Lambda}\|_1-\|{P} {\Lambda}P\|_1\}=\frac{1}{4}\sum_{k=m}^{m^2-1}|\Lambda^{\downarrow}_{k}|,
\end{eqnarray}
A comparison of Eq.  (\ref{NewGD2}) with Eq. (\ref{Q-Rho-oneNorm}) shows that they are, in general, different except for some special cases. More precisely,    $[D_{{\mathcal T},f=1}^{(2)}(\rho)]^2$ and $Q_{\Lambda}(\rho)$ are obtained from the first $m(m-1)$ smaller eigenvalues of the matrix $\mathcal{T}_f\mathcal{T}_f^\T$, with the pair $\{f_1,f_2\}$ given by $\{f_1=1,f_2=1\}$ and $\{f_1=i\sqrt{\frac{2}{n}}y,f_2=1\}$, respectively.
Evidently when  $\vec{x}=0$, the two measures are completely equivalent. On the other hand,  $Q_{\Lambda}(\rho)$ gives the same result for all states with $y=0$ and different $\vec{x}$, i.e. it becomes independent of the coherence vector $\vec{x}$  whenever $y=0$, but $[D_{{\mathcal T},f=1}^{(2)}(\rho)]^2$ preserve the rule of the coherence vector $\vec{x}$  in this case.

{\it One-norm geometric discord.---}
Interestingly, for $f=1$ the operator norm distance   coincides with the one-norm geometric quantum discord \cite{Nakano2013,Paula2013,Ciccarello2014},  for some special two-qubit cases. This happens, for instance,  for Bell-diagonal states
for which $\vec{x}=0$ and $T=\diag\{t_1,t_2,t_3\}$. In this case we find that  $D_{{\mathcal T},f=1}^{(\infty)}(\rho)=|t_2|/2$, where we have supposed that $|t_1|\ge |t_2|\ge |t_3|$. On the other hand for states with $\vec{x}\neq 0$, $T=\diag\{t,t,t\}$ our definition gives  $D_{{\mathcal T},f=1}^{(\infty)}(\rho)=|t|/2$ \cite{Paula2013,Ciccarello2014}.

We give bellow some illustrative examples.
\subsection{Examples}

{\it $m\otimes m$ Werner states.---}  For the $m\otimes m$ Werner states
\begin{eqnarray}
\rho=\frac{m-x}{m^3-m}\Id_m+\frac{m x-1}{m^3-m}F, \qquad x\in [-1,1],
\end{eqnarray}
with $F=\sum_{k,l=1}^{m}\ket{kl}\bra{lk}$, the geometric measure of discord is \cite{LuoFu2010}
\begin{eqnarray}
D_G(\rho)=\frac{(m x-1)^2}{m(m-1)(m+1)^2}.
\end{eqnarray}
On the other hand for these states  $\vec{x}=\vec{y}=0$ and
\begin{eqnarray}
{\mathcal T}{\mathcal T}^{\T}=\diag\{\tau,\cdots,\tau\},\;\; \textrm{with}\;\; \tau=\frac{(mx-1)^2}{m^2(m^2-1)^2},
\end{eqnarray}
so that we get $[D_{{\mathcal T},f=1}^{(2)}(\rho)]^2=m(m-1)\tau=D_G(\rho)=[D_{{\mathcal T},\mu}^{(2)}(\rho)]^2/m$.

{\it $m \otimes m$ Isotropic States.---} As the second example we consider $m\otimes m$ isotropic states defined by
\begin{eqnarray}
\rho=\frac{1-x}{m^2-1}\Id_m+\frac{m^2 x-1}{m^2-1}\ket{\psi}\bra{\psi}, \quad x\in [0,1],
\end{eqnarray}
with $\ket{\psi}=\frac{1}{\sqrt{m}}\sum_{k=1}^{m}\ket{kk}$.
The geometric measure of discord is \cite{LuoFu2010}
\begin{eqnarray}
D_G(\rho)=\frac{(m^2 x-1)^2}{m(m-1)(m+1)^2}.
\end{eqnarray}
On the other hand for these states  $\vec{x}=\vec{y}=0$ and
\begin{eqnarray}
{\mathcal T}{\mathcal T}^{\T}=\diag\{\tau,\cdots,\tau\},\;\; \textrm{with}\;\;  \tau=\frac{(m^2x-1)^2}{m^2(m^2-1)^2},
\end{eqnarray}
we obtain $[D_{{\mathcal T},f=1}^{(2)}(\rho)]^2=m(m-1)\tau=D_G(\rho)=[D_{{\mathcal T},\mu}^{(2)}(\rho)]^2/m$.

{\it Pure $m\otimes m$ states.---}
Next, we consider  an example of bipartite $m\otimes m$ pure state $\ket{\Psi}$,  with the following Schmidt decomposition
\begin{equation}
\ket{\Psi}=\sum_{i=1}^m\sqrt{s_i}\ket{i}\ket{i}.
\end{equation}
The geometric discord of this state is \cite{LuoFuPRL2011,LuoFu2012}
\begin{equation}\label{GDPuremxm}
D_G(\Psi)=1-\sum_{i=1}^{m}s_i^2=1-\Tr{(\rho^A)^2}=\frac{1}{2}C^2(\Psi),
\end{equation}
where $\rho^A=\Tr_{B}{\project{\Psi}}$ is the reduced state of the subsystem $A$, and  $C(\Psi)$ is the generalized concurrence of $\ket{\Psi}$ \cite{Rungta2001}. On the other hand, in order to evaluate $D_{{\mathcal T},f}^{(p)}(\Psi)$ we have to find the local coherence vectors and the correlation matrix associated with $\rho=\project{\Psi}$, we get
\begin{eqnarray}
x_k&=&y_k=\frac{m}{2}\sum_{i=1}^ms_i\bra{i}\hat{\lambda}_k\ket{i}, \\ \nonumber
t_{kl}&=&\frac{m^2}{4}\sum_{i=1}^m \sum_{j=1}^m\sqrt{s_is_j}\bra{i}\hat{\lambda}_k\ket{i}\bra{i}\hat{\lambda}_l^\ast\ket{i}, \\
&=&\frac{m^2}{4}\Tr{\left(\sqrt{\rho^A}\hat{\lambda}_k \sqrt{\rho^A}\hat{\lambda}_l^\ast\right)},
\end{eqnarray}
for $k,l=1,\cdots, m^2-1$, where  $\{\hat{\lambda}_k\}_{k=1}^{m^2-1}$ are basis of $SU(m)$ algebra. If we choose the basis of $SU(m)$ in such a way that the first $m-1$ generators make the basis of its Cartan subalgebra \cite{GeorgiBook1999}, we get
\begin{eqnarray}
x_k&=&y_k=\left\{\begin{array}{ll}\frac{m\left(\sum_{i=1}^k s_i-ks_{k+1}\right)}{\sqrt{2k(k+1)}}, & \;k=1,\cdots, m-1
\\ 0  & k=m,\cdots,m^2-1\end{array}\right.
 \\
T&=&\frac{m^2}{2}\left(\begin{array}{c|c}T_c & 0  \\ \hline
0 & T_d
\end{array}\right),
\end{eqnarray}
where $T_c$  is an $(m-1)\times (m-1)$ symmetric matrix with
\begin{eqnarray}
(T_c)_{kk}&=&\frac{1}{k(k+1)}\left(\sum_{i=1}^k s_i+k^2 s_{k+1}\right), \\
(T_c)_{k<l}&=&\frac{1}{\sqrt{k(k+1)l(l+1)}}\left(\sum_{i=1}^k s_i-k s_{k+1}\right),
\end{eqnarray}
and $T_d$ is an $m(m-1)\times m(m-1)$ diagonal matrix
such that $T_d=\diag\{\pm\sqrt{s_{1}s_{2}},\pm\sqrt{s_{1}s_{3}},\cdots,\pm\sqrt{s_{m^2-2}s_{m^2-1}}\}$.
To continue, we have to calculate  eigenvalues of the $(m^2-1)\times (m^2-1)$-dimensional  matrix ${\mathcal T}{\mathcal T}^\T$ which, except for $m=2$, does not have simple form in general. For instance,  for $m=3$ we find ${\mathcal T}{\mathcal T}^\T=\diag\{\tau_{+},\tau_{-},s_{1}s_{2},s_{1}s_{2},s_{1}s_{3},s_{1}s_{3},s_{2}s_{3},s_{2}s_{3}\}$, with
$\tau_{\pm}=\frac{1}{3} \left(\sum_{i=1}^3s_{i}^2\pm\sqrt{\sum_{i=1}^3s_{i}^4-s_1^2s_2^2-s_1^2s_3^2-s_2^2s_3^2}\right)$, which can be used to evaluate $D_{{\mathcal T},f}^{(p)}(\Psi)$. In this case, we have plotted  $D_{{\mathcal T},f}^{(p)}(\Psi)$  for $s=1,2$ and  $f=1$ and $f=\mu$ (see Fig. \ref{Fig3DPure}). For comparison,  square of the normalized geometric discord, namely concurrence $C(\psi)$ given by Eq. (\ref{GDPuremxm}), is also plotted.  As it is evident from this figure, $D_{{\mathcal T},f}^{(p)}(\Psi)$ is a  monotone function of $C(\psi)$, so that it can be regarded as  a  measure of entanglement for pure states.     On the other hand, for  the  maximally entangled states of arbitrary $m$ we have  $s_i=\frac{1}{m}$ for $i=1,\cdots m$, leads to $\vec{x}=\vec{y}=\vec{0}$, $TT^{\T}=\frac{m^2}{4}I_{m^2-1}$; so that ${\mathcal T}{\mathcal T}^{\T}=\frac{1}{m^2}I_{m^2-1}$ and  $[D_{{\mathcal T},f=1}^{(2)}(\Psi)]^2=\frac{m-1}{m}=D_G(\Psi)=[D_{{\mathcal T},\mu}^{(2)}(\Psi)]^2/m$.

\begin{figure}[ht!]
\centering
\includegraphics[width=18cm]{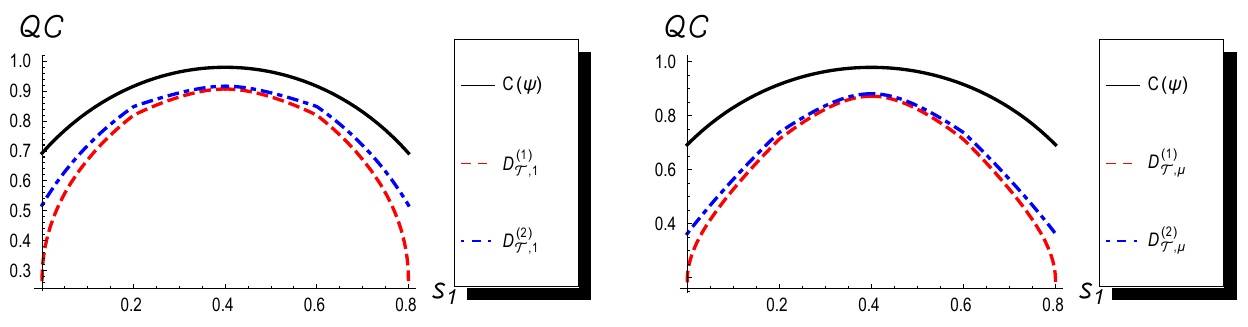}
\caption{(Color online) Quantum correlations of a general  pure state $\ket{\psi}$ (for $m=3$) \textit{vs.}  $s_1$ with $s_2=0.2$. Two special cases  $f=1$ (left) and $f=\mu$ (right) are considered. For a comparison, square of the normalized geometric discord (which is equal to the concurrence $C(\psi)$) is also plotted. All measures are normalized}‎
\label{Fig3DPure}
\end{figure}

{\it A two-parameter class of two-qubit states.---} As an another illustrative example let us consider a two-parameter class of the two-qubit X-states discussed in  \cite{AlQasimiPRA2011}
\begin{eqnarray}\label{Rho-ab}
\rho=\frac{1}{2}\left(\begin{array}{cccc}a & 0 & 0 & a \\ 0 & 1-a-b & 0 & 0 \\ 0 & 0 & 1-a+b & 0 \\ a & 0 & 0 & a\end{array}\right),
\end{eqnarray}
where $0\le a \le 1$ and $a-1 \le b \le 1-a$. The quantum discord of this state is \cite{AlQasimiPRA2011}
\begin{equation}\label{Q-Rho-ab}
Q(\rho)=\min\{a, q\},
\end{equation}
where
\begin{eqnarray}\nonumber
q &=&\frac{a}{2}\log_{2}{\left[\frac{4a^2}{(1-a)^2-b^2}\right]}-\frac{b}{2}\log_{2}{\left[\frac{(1+b)(1-a-b)}{(1-b)(1-a+b)}\right]}\\ \nonumber \label{q}
&-&\frac{\sqrt{a^2+b^2}}{2}\log_{2}{\left[\frac{1+\sqrt{a^2+b^2}}{1-\sqrt{a^2+b^2}}\right]}
+\frac{1}{2}\log_{2}{\left[\frac{4((1-a)^2-b^2)}{(1-b^2)(1-a^2-b^2)}\right]}. \end{eqnarray}
For this state we get ${\mathcal T}{\mathcal T}^\T=\frac{1}{4}\diag(a^2,a^2,(1-2a)^2+b^2)$ and $\mu(\rho^B)=\half (1+b^2)$. Figure (\ref{twoqubit}) compare the behavior of the above geometric measures of quantumness with the quantum discord.

\begin{figure}[ht!]
\centering
\includegraphics[width=18cm]{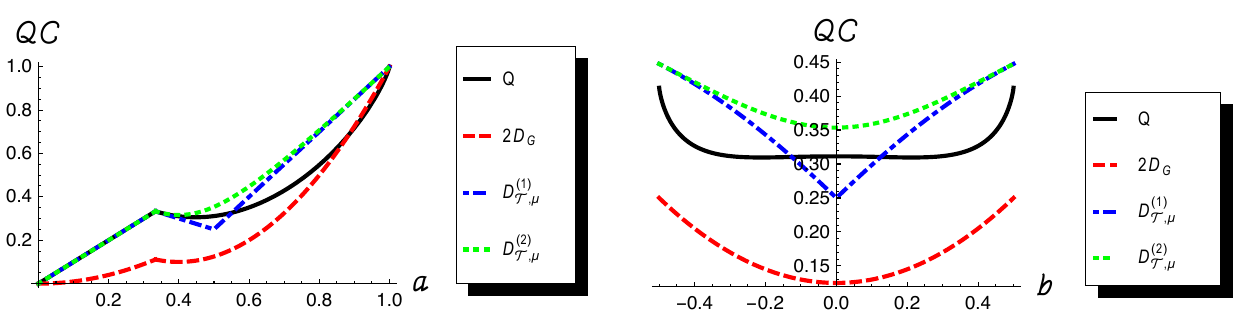}
\caption{(Color online) Quantum correlations \textit{vs.} $a$ for $b=0$ (left), and \textit{vs.} $b$ for $a=0.5$ (right). All measures are normalized.}‎
\label{twoqubit}
\end{figure}

\section{Conclusion}
We have presented a class of computable quantifiers of the quantum correlation for an arbitrary bipartite state. Our measures  are based on the necessary and sufficient condition for a state to be zero-discord. The analytical expression for these measures are  given for any bipartite state. Interestingly, we have shown that this class of measures includes the tight lower bound on the geometric discord, so that this lower bound can be used as an independent indicator  of the quantumness of correlation. We have also introduced a measure of the quantum correlation which is invariant under local quantum channels performed on the unmeasured subsystem. It is also shown that a way to prevent the geometric measure from
increasing under local operations on the unmeasured subsystem is to divide it by the purity of this subsystem. We have provided some examples and exemplified  our measure.

\section*{Acknowledgments}
The authors wish to thank The Office of Graduate
Studies of The University of Isfahan for their support.

\appendix
\section{Geometric discord and its tight lower bound}\label{AppendixTight}
 In this appendix we provide a proof to show that  Eq. (\ref{GD-AlternativeForm2-0}) can be regarded as an alternative form for the geometric discord (\ref{GD-Dakic}).

{\it Alternative form for geometric discord.---}
Let  $\{\ket{k}\}_{k=1}^{m}$ be any orthonormal base for $\mathcal{H}^{A}$.  Following \cite{LuoFu2010}  we represent the projection operators corresponding to  this base as
\begin{eqnarray}\label{Pi-HS1}
\Pi_{k}^{A}=\ket{k}\bra{k}=\sum_{i=0}^{m^2-1}a_{ki}X_{i},
\end{eqnarray}
where $a_{ki}$ are defined in Eq. (\ref{aki}) and $k=1,\cdots,m$. It is easy to see that we can write matrix $A=(a_{ki})$ as below
\begin{equation}
A=\left(\begin{array}{cc}\frac{1}{\sqrt{m}} & \vec{\mu}_1^{\T} \\
\vdots & \vdots \\
\frac{1}{\sqrt{m}} & \vec{\mu}_{m}^{\T}\end{array}\right),
\end{equation}
where $\vec{\mu}_k=\frac{\sqrt{2}}{m}\vec{\alpha}_k$ with $\vec{\alpha}_k$ as defined in Eqs. (\ref{CVR-Pure2}) and (\ref{CVR-Pure3}). Therefore vectors $\{\vec{\mu}_k\}_{k=1}^{m}$ make the  $(m-1)$-dimensional simplex $\Delta^{m-1}_{\{\vec{\mu}_k\}\in \mathbb{R}^{m^2-1}}$.
Using Eq. (\ref{C}) we get
\begin{eqnarray}\label{CCt}
\Tr{(CC^T)}=\frac{1}{mn}\left[\left(1+\frac{2}{n}\vec{y}^{\T}\vec{y}\right) +\frac{2}{m}\Tr{G}\right],
\end{eqnarray}
and
\begin{eqnarray}
[ACC^TA^T]_{kk^{\prime}}=
\frac{1}{m^2n}\left[\left(1+\frac{2}{n}\vec{y}^{\T}\vec{y}\right)+2\vec{\mu}_k^{\T} G\vec{\mu}_{k^{\prime}}
\right]
+\frac{\sqrt{2}}{m^2n}\left[\vec{\mu}_k^{\T}\left(\vec{x}+\frac{2}{n}T\vec{y}\right)+\left(\vec{x}^{\T}+\frac{2}{n}\vec{y}^{\T} T^{\T}\right)\vec{\mu}_{k^{\prime}}\right],
\end{eqnarray}
where $G$ is defined by Eq. (\ref{G}). We find therefore
\begin{eqnarray}\label{ACCtAt}
\Tr{[ACC^TA^T]}=\frac{1}{mn}\left[\left(1+\frac{2}{n}\vec{y}^{\T}\vec{y}\right)  +\frac{2}{m}\sum_{k=1}^{m}\vec{\mu}_k^{\T}G\vec{\mu}_k
\right],
\end{eqnarray}
where we have used the fact that $\sum_{k=1}^{m}\vec{\mu}_k=\vec{0}$.
Substituting Eqs. (\ref{CCt}) and (\ref{ACCtAt}) into Eq. (\ref{GD-Luo1}), we arrive at the following form for the geometric discord
\begin{eqnarray}\label{GD-AlternativeForm2}
D_{G}(\rho)=\frac{2}{m^2 n}\left[\Tr{G}-\max_{\{\vec{\mu}_k\}} \sum_{k=1}^{m}\vec{\mu}_k^{\T} G \vec{\mu}_k\right].
\end{eqnarray}
Here maximum is taken over all simplexes $\Delta^{m-1}_{\{\vec{\mu}_k\}\in \mathbb{R}^{m^2-1}}$, i.e. over  all vectors $\{\vec{\mu}_k\}_{k=1}^{m}\in \mathbb{R}^{m^2-1}$ fulfilling conditions
 $\vec{\mu}_k\cdot\vec{\mu}_{k^\prime}=\left(\delta_{kk^\prime}-\frac{1}{m}\right)$ and $\sum_{k=1}^{m}\vec{\mu}_k=\vec{0}$.
To gain further insight into the meaning of the above equation, it is worth to compare it with Eq.  (\ref{NewGD2}) for $s=2$, $f_1(y)=f_2(y)=1$.
It turns out that   the  calculation of $[D_{{\mathcal T},f=1}^{(2)}(\rho)]^2$ needs   to perform optimization over $(m-1)$-dimensional projection operators $P$, which can be solved exactly, but in calculation of $D_G(\rho)$ we have to make optimization over $(m-1)$-dimensional simplexes $\Delta^{m-1}_{\{\vec{\mu}_k\}\in {\mathbb R}^{m^2-1}}$, where does not have an exact solution in general. Two definitions become identical when $m=2$, namely for $2\otimes n$ systems. This happens because in case $m=2$, calculation of the geometric discord leads  to the problem of optimization over one-dimensional simplexes $\Delta^{1}_{\{\vec{\mu}_1,\vec{\mu}_2\}\in {\mathbb R}^{3}}$ with $\vec{\mu}_1=-\vec{\mu}_2=\frac{1}{\sqrt{2}}\vec{\alpha}_1$ and $|\vec{\alpha}_1|=1$, which is the same as the problem of optimization over one-dimensional projection operators $P$, and  get
\begin{eqnarray}\label{GD-m=2}
[D_{{\mathcal T},f=1}^{(2)}(\rho)]^2=D_{G}(\rho)=\frac{1}{2 n}\left[\Tr{G}-\max_{\vec{\alpha}_1}\{\vec{\alpha}_1^{\T} G \vec{\alpha}_1\}\right]
=\frac{1}{2 n}\left[\Tr{G}-\eta_{1}\right]=\frac{1}{2 n}\left[\eta_2+\eta_3\right].
\end{eqnarray}
where we have defined $\eta_{1}\ge \eta_{2}\ge \eta_{3}\ge 0$ as the eigenvalues of  $G$. This agrees with the result obtained in Refs. \cite{LuoFuPRL2011,Saj2012}.

{\it Tight lower bound on the geometric discord \cite{Rana2012,Hassan2012}.---}
Unfortunately, for $m>2$,  the maximization involved in Eq. (\ref{GD-AlternativeForm2}) can not be solved analytically and we need to obtain lower bound. To do so,  let  $\{\ket{s}\}_{s=1}^{m}$ be the standard base of the  space $\mathcal{H}^{A}$, namely the one which the $SU(m)$ generators $\{\hat{\lambda}^A_i\}_{i=1}^{m^2-1}$ are expanded in terms of them.
Similar to Eq. (\ref{Pi-HS1}), we can write
\begin{eqnarray}
\Pi_{s}^{A}=\ket{s}\bra{s}=\sum_{i=0}^{m^2-1}b_{si}X_{i},
\end{eqnarray}
where
\begin{eqnarray}
b_{si}=\Tr{[\ket{s}\bra{s}X_{i}]}=\bra{s}X_{i}\ket{s},
\end{eqnarray}
for $s=1,\cdots,m$ and $i=0,\cdots,m^2-1$.
Now if we choose the basis of the algebra in such a way that Cartan subalgebra makes the first $m-1$ generators, then we can write matrix $B=(b_{si})$ as follows
\begin{equation}\label{BMatrix}
B=\left(\begin{array}{cc}\frac{1}{\sqrt{m}} & \vec{{\tilde \nu}}_1^{\T}  \\
\vdots & \vdots  \\
\frac{1}{\sqrt{m}} & \vec{{\tilde \nu}}_{m}^{\T} \end{array}\right).
\end{equation}
Here $\{\vec{{\tilde \nu}}_s\}_{s=1}^{m}$ are vectors in ${\mathbb R}^{m^2-1}$ such that only first $m-1$ components of them are nonzero. So, we can write  $\vec{{\tilde \nu}}_s=(\vec{\nu}_s,\vec{0})$ where $\{\vec{\nu}_s\}_{s=1}^{m}$ are vectors in ${\mathbb R}^{m-1}$, and $\vec{0}$ denotes null vectors in ${\mathbb R}^{m(m-1)}$. It is worth to mention that  vectors $\{\vec{\nu}_s\}_{s=1}^{m}$ are in fact weight vectors of the $SU(m)$ Lie algebra in the defining representation \cite{GeorgiBook1999} and satisfy the following orthonormality condition
\begin{equation}\label{Nu-Orthonormality}
\sum_{s=1}^{m}(\vec{\nu}_s)_k(\vec{\nu}_s)_l=\delta_{kl}.
\end{equation}
In view of this, the zero vectors $\vec{0}$ of the definition $\vec{{\tilde \nu}}_s=(\vec{\nu}_s,\vec{0})$ arise from  the diagonal elements of the  root operators of the algebra, which are all zero. Therefore vectors $\{\vec{\tilde{\nu}}_s\}_{s=1}^{m}$ makes simplex $\Delta^{m-1}_{\{\vec{\tilde{\nu}}_s\}\in \mathbb{R}^{m^2-1}}$, or equivalently  simplex $\Delta^{m-1}_{\{\vec{\nu}_s\}\in \mathbb{R}^{m-1}}$.
Evidently,  the general base $\{\ket{k}\}_{k=1}^{m}$ can be obtained from the standard one by a unitary transformation $U\in SU(m)$ as
$\{\ket{k}\}=U\{\ket{s}\}$. Corresponding to this, there exists orthogonal transformation $\tilde{R}\in SO(m^2-1)$ such that the general simplex $\Delta^{m-1}_{\{\vec{\mu}_k\}\in \mathbb{R}^{m^2-1}}$ can be obtained from $\Delta^{m-1}_{\{\vec{\tilde{\nu}}_s\}\in \mathbb{R}^{m^2-1}}$,  i.e.
\begin{eqnarray}
(\vec{\mu}_k)_i=\sum_{j=1}^{m^2-1}{\tilde R}_{ij}(\vec{\tilde{\nu}}_k)_j=\sum_{j=1}^{m-1}R_{ij}(\vec{\nu}_k)_j,
\end{eqnarray}
for $i=1,2,\cdots, m^2-1$. In the second equality $R=(R_{ij})=(\hat{n}_j)_i$ is an $(m^2-1)\times(m-1)$ left orthogonal matrix \cite{Rana2012,Hassan2012}, i.e. $R^{\T}R=I_{m-1}$, and $\hat{n}_j\in \mathbb{R}^{m^2-1}$ ($j=1,\cdots,m-1$) are orthonormal vectors, i.e. $\hat{n}_i\cdot \hat{n}_{i^\prime}=\delta_{ii^\prime}$. Using this and Eq. (\ref{Nu-Orthonormality}), we get
\begin{eqnarray}\label{etai}
\max_{\{\vec{\mu}_k\}} \sum_{k=1}^{m}\vec{\mu}_k^{\T} G \vec{\mu}_k\le \max_{\{\hat{n}_i\}} \sum_{i=1}^{m-1}\hat{n}_i^{\T} G \hat{n}_i =  \sum_{i=1}^{m-1}\eta_{i}^\downarrow,
\end{eqnarray}
where $\{\eta_k^{\downarrow}\}_{k=1}^{m^2-1}$ are eigenvalues of $G$ in nonincreasing order.
Using this in Eq. (\ref{GD-AlternativeForm2}), we find the  desired lower bound (\ref{TightLowerBound}) for the geometric discord, which is already obtained in Refs. \cite{Rana2012,Hassan2012}. It is worth to mention that in the particular case  $m=2$, the obtained bound gives exact result for the geometric discord (see Eq. (\ref{GD-m=2})). This follows from the homomorphism $SU(2)\sim SO(3)$, happens only for $m=2$.
 On the other hand, for $m>2$ the set of all unitary transformations $U\in SU(m)$ acting on the $m$-dimensional Hilbert space $\mathcal{H}^A$ will be a subset of the matrices in $SO(m^2-1)$. This implies that there exist rotations $\tilde{R}\in SO(m^2-1)$ that are not correspond to any $U\in SU(m)$, leading therefore to the inequality (\ref{etai}).

\section{A proof for Theorem \ref{TheoremZDS-NC}}\label{AppendixProof}
In this appendix we provide a proof for theorem \ref{TheoremZDS-NC}. To this aim, we need the following lemma.

\begin{lemma}\label{Lemma-ZDS-CVR-xyT}
 (i) If $\rho$ is a zero-discord state on the space $\mathcal{H}^{A}\otimes\mathcal{H}^{B}$, then its corresponding local coherence vectors $\vec{x}$, $\vec{y}$, and the correlation matrix $T$ can be represented by the following equations
\begin{eqnarray}\label{ZDS-CVR-xy}
\vec{x}&=&\sum_{k=1}^{m}p_k \vec{\alpha}_{k},\qquad \vec{y}=\sum_{k=1}^{m}p_k \vec{\xi}_{k}, \\ \label{ZDS-CVR-T}
T&=&\sum_{k=1}^{m}p_k (\vec{\alpha}_{k})(\vec{\xi}_{k})^{\T},
\end{eqnarray}
where $\{\vec{\alpha}_{k}\}_{k=1}^m$ denote coherence vectors associated to the orthonormal projection operators of the subsystem $A$, hence satisfy Eqs. (\ref{CVR-Pure2}) and (\ref{CVR-Pure3}), but $\{\vec{\xi}_{k}\}_{k=1}^m$  are coherence vectors of arbitrary states of the subsystem $B$.

(ii) If $\rho$ is an arbitrary bipartite state, then its corresponding local coherence vectors $\vec{x}$ and $\vec{y}$ can be represented by Eq. (\ref{ZDS-CVR-xy}).
\end{lemma}
\begin{proof}
(i) Use  the coherence vector representations for $\Pi_k^A$ and $\rho_k^B$ as
\begin{eqnarray}
\Pi_k^{A}=\frac{1}{m}\left(\Id+\vec{\alpha}_{k}\cdot \hat{\lambda}^A\right),\quad
\rho_k^{B}=\frac{1}{n}\left(\Id+\vec{\xi}_{k}\cdot \hat{\lambda}^B\right),
\end{eqnarray}
and insert them in the definition of zero-discord state (\ref{ZDS}). Comparing the result with the definition of $\rho$ given in Eq. (\ref{Rho-Bipart2}), one can obtain the coherence vectors $\vec{x}$, $\vec{y}$ and the correlation matrix $T$ as given by Eqs. (\ref{ZDS-CVR-xy}) and (\ref{ZDS-CVR-T}).

(ii) Let $\rho^{A}=\sum_{k=1}^{m}p_k\Pi_k^A$, with $\{\Pi_k^A\}_{k=1}^{m}$ orthonormal projections on $\mathcal{H}^{A}$, be the eigenspectral decomposition of $\rho^{A}$. Then denoting coherence vectors of $\{\Pi_k^A\}_{k=1}^{m}$ by $\{\vec{\alpha}_k\}_{k=1}^{m}$, we find that  $\vec{x}=\sum_{k=1}^{m}p_k \vec{\alpha}_{k}$. Now having $\{p_k\}_{k=1}^{m}$, we can always find set $\{\rho_k^B\}_{k=1}^{m}$ such that ensemble $\{p_k,\rho_k^B\}_{k=1}^{m}$ realizes $\rho^B$, i.e. $\rho^B=\sum_{k=1}^{m}p_k\rho_k^B$. Now letting $\{\vec{\xi}_k\}_{k=1}^{m}$ be coherence vectors of $\{\rho_k^B\}_{k=1}^{m}$, we get  $\vec{y}=\sum_{k=1}^mp_k\vec{\xi}_k$. Note that for a given probability set $\{p_k\}_{k=1}^{m}$, states $\{\rho_k^B\}_{k=1}^{m}$ which realize  $\rho^B$ are not unique,  so  associated coherence vectors  $\{\vec{\xi}_k\}_{k=1}^{m}$ are not unique too.
\end{proof}
Now we are in a position to present the proof for theorem \ref{TheoremZDS-NC}. If $\rho$ is a zero-discord state, then by lemma \ref{Lemma-ZDS-CVR-xyT} its corresponding local coherence vectors $\vec{x}$, $\vec{y}$ and correlation matrix $T$ can be represented by Eqs. (\ref{ZDS-CVR-xy}) and (\ref{ZDS-CVR-T}), with $\{\vec{\alpha}_k\}_{k=1}^{m}$ as coherence vectors corresponding to orthonormal projections. Defining ${P}$ as (\ref{Projection}) and using the properties $\{\vec{\alpha}_k\}_{k=1}^{m}$ given in Eq.  (\ref{CVR-Pure2}), one can easily shows that conditions (\ref{ZDS-Pi-xT}) are satisfied.
Conversely, we have to proof that if Eq.  (\ref{ZDS-Pi-xT}) is satisfied, then $\rho$ is a zero-discord state, i.e. its corresponding $\vec{x}$, $\vec{y}$ and $T$ have the form given by Eqs. (\ref{ZDS-CVR-xy}) and (\ref{ZDS-CVR-T}).
To do this, we first note that Eq. (\ref{ZDS-CVR-xy}) is satisfied for a general state $\rho$.
But by assumption Eq. (\ref{ZDS-Pi-xT}) is also satisfied, leading therefore  to the following form for the correlation matrix $T$
\begin{eqnarray}\label{T-p-alpha-eta}
T=\sum_{k=1}^{m}\sum_{l=1}^{m}p_{kl}(\vec{\alpha}_{k})(\vec{\eta}_{l})^{\T}.
\end{eqnarray}
Since $\{\vec{\xi}_i\}_{i=1}^m$ are not unique, we can therefore choose them in such a way that  they can be expanded in terms of $\{\vec{\eta}_l\}_{l=1}^m$ as
$p_k\vec{\xi}_k=\sum_{l=1}^m p_{kl}\vec{\eta}_l$. Substituting this into Eq. (\ref{T-p-alpha-eta}) we get Eq. (\ref{ZDS-CVR-T}), therefore   $\vec{x}$, $\vec{y}$  and $T$ take  the form given by Eqs.  (\ref{ZDS-CVR-xy}) and (\ref{ZDS-CVR-T}), hence $\rho$ is a zero-discord state.



\end{document}